\documentclass[a4paper,11pt]{article}

\usepackage{fullpage}
\usepackage{times}
\usepackage{soul}
\usepackage{url}
\usepackage[utf8]{inputenc}
\usepackage[small]{caption}
\usepackage{graphicx}
\usepackage{xcolor}
\usepackage{amsmath}
\usepackage{booktabs}
\usepackage{algorithm}
\usepackage{enumitem}

\usepackage{amsthm}
\usepackage{amssymb}
\usepackage{algpseudocode}

\usepackage{bbm}

\newtheorem{theorem}{Theorem}[section]

\newtheorem{lemma}[theorem]{Lemma}

\theoremstyle{definition}

\algnotext{EndFor}
\algnotext{EndIf}
\algnotext{EndWhile}
\algnotext{EndProcedure}

\usepackage{natbib}

\newcount\Comments  
\newcommand{\kibitz}[2]{\ifnum\Comments=1{\color{#1}{#2}}\fi}

\allowdisplaybreaks

\usepackage{authblk}

\title{\bf Envy-Freeness in House Allocation Problems}

\author{Jiarui Gan}
\author{Warut Suksompong}
\author{Alexandros A. Voudouris}
\affil{Department of Computer Science, University of Oxford}

\begin{document}

\maketitle

\begin{abstract}
We consider the house allocation problem, where $m$ houses are to be assigned to $n$ agents so that each agent gets exactly one house. 
We present a polynomial-time algorithm that determines whether an envy-free assignment exists, and if so, computes one such assignment. We also show that an envy-free assignment exists with high probability if the number of houses exceeds the number of agents by a logarithmic factor.
\end{abstract}

\section{Introduction} \label{sec:intro}

In the \emph{house allocation problem}, also known as the \emph{assignment problem}, a set of $m$ houses are to be assigned to a set of $n$ agents with preferences over the houses, under the constraint that each agent is assigned exactly one house \citep{HyllandZe79,Zhou90,AbdulkadirogluSo98}.
Some economic efficiency condition is often desired, for example that the assignment is \emph{Pareto optimal}. This means that no other assignment makes some agent better off and no agent worse off in comparison to the current assignment \citep{AbrahamCeMa04,Manlove13}. 

In this note, we investigate the issue of fairness in house allocation using the well-established fairness notion of \emph{envy-freeness} \citep{Foley67,Varian74}. 
An allocation is said to be \emph{envy-free} if every agent likes her house at least as much as any other assigned house. 
Clearly, an envy-free allocation does not always exist, for example when all agents have the same strict ranking over the houses.
If the number of agents is equal to the number of houses, then all houses must be assigned.
In this case, it is easy to see that determining whether an envy-free assignment exists, and computing one if so, can be done in polynomial time. 
Indeed, we can simply construct a bipartite graph with the agents on one side and the houses on the other side, and add an edge between an agent and a house whenever the agent likes the house at least as much as any other house. 
An envy-free assignment exists if and only if the graph admits a perfect matching; it is well-known that the latter condition can be checked in polynomial time.

The purpose of our note is to study envy-freeness in the general house allocation problem where the number of houses can exceed the number of agents.
Formally, there are $m$ houses $M=\{1,2,\dots, m\}$ and $n$ agents $N = \{1,2,\dots,n\}$, where $m\geq n$. Each agent has a ranking over the houses, where ties are permitted. 
Allowing the number of agents and the number of houses to be different makes the problem more complex, and we can no longer determine the existence of envy-free assignments solely by matching agents to their favorite houses.
For example, if there are three houses and two agents with the rankings $1\succ 2\succ 3$ and $1\succ 3\succ 2$ over the houses, then even though both agents compete for the same top house, there is an envy-free assignment that assigns house 2 to agent 1 and house 3 to agent 2.
Nevertheless, we present a polynomial-time algorithm that determines whether an envy-free assignment exists, and computes one if it does. 
We then show that if the number of houses exceeds the number of agents by a logarithmic factor, an envy-free assignment exists with high probability.

To the best of our knowledge, the only work before ours to have considered envy-freeness in house allocation is that of \cite{BeynierChGo18}. 
Their work focuses exclusively on the $m=n$ case but contains the extra feature that agents are placed on a network that describes the envy relation, and they showed algorithms and hardness results for different networks.
Recently, \cite{Segalhalevi19} studied a concept called envy-free matchings on bipartite graphs, and provided conditions under which a non-empty envy-free matching exists along with algorithms to compute such matchings. In contrast to this note, his study is restricted to unweighted bipartite graphs, which correspond to each agent either approving or disapproving each house, and does not require every agent to be assigned to a house.

\section{Our Results}

Denote by $G=(X,Y,E)$ a bipartite graph with bipartite vertex sets $X,Y$ and edge set $E$. 
For any set of vertices $V$, denote by $S(V)$ the set of vertices that are adjacent to at least one vertex in $V$.
An \emph{X-saturating} matching is a matching that covers every vertex in $X$.
A set $Z\subseteq X$ is said to be a \emph{Hall violator} if $|Z|>|S(Z)|$. It is said to be a \emph{minimal Hall violator} if no $Z'\subset Z$ is a Hall violator.
Recall that by Hall's Theorem, an $X$-saturating matching exists if and only if $|Z|\leq |S(Z)|$ for all $Z\subseteq X$. 
In other words, there is an $X$-saturating matching exactly when no Hall violator is present.

As part of our algorithm, we will need to find a minimal Hall violator in the case where no $X$-saturating matching exists.
In particular, we show that if there is a Hall violator, it is possible to find a minimal one efficiently. 
Our approach is similar to that in Lemma~4.5 of \cite{AmanatidisMaNi17}.

\begin{lemma}
\label{lem:hall-violator}
Given a bipartite graph $G=(X,Y,E)$ without an $X$-saturating matching, a minimal Hall violator can be found in polynomial time.
\end{lemma}

\begin{proof}
Let $B$ be a maximum matching of $G$, and let $X_m$ and $X_u$ be the set of vertices in $X$ that are matched and unmatched in $B$, respectively. 
Since $G$ does not admit an $X$-saturating matching, $|X_u|>0$.
Let $z$ be an arbitrary vertex in $X_u$.
Construct an auxiliary directed graph $G'$ with the same vertex set as $G$ as follows. For every edge $(x,y)\in E$ with $x\in X$ and $y\in Y$, add a directed edge from $x$ to $y$ in $G'$. In addition, for every edge $(x,y)\in B$ with $x\in X$ and $y\in Y$, add a directed edge from $y$ to $x$ in $G'$.
Let $Z$ be the set of vertices reachable from $z$ in $G'$. We claim that $Z$ is a minimal Hall violator. Note that $Z$ can be computed efficiently using depth-first search.

First, we show that $Z$ is a Hall violator, i.e., $|Z|>|S(Z)|$. Every vertex in $S(Z)$ is reachable from $z$ in $G'$. 
If a vertex $v\in S(Z)$ is unmatched in $B$, then by construction, a path from $z$ to $v$ alternates between edges in $B$ and edges not in $B$, starting and ending with edges not in $B$. 
Since $z$ and $v$ are not matched in $B$, this path is an augmenting path, contradicting the maximality of $B$. So every vertex in $S(Z)$ is matched in $B$, implying the existence of an injection from $S(Z)$ to $Z$. 
Since $z\in X_u$, this injection is not a surjection. It follows that $|Z|>|S(Z)|$. Observe also that every vertex in $Z$ besides $z$ is matched in $B$ by construction, so in fact we have $|Z|=|S(Z)|+1$.

Next, we show that there is no $Z'\subset Z$ such that $|Z'|>|S(Z')|$. If $z\not\in Z'$, then since all vertices in $Z'$ are matched in $B$, we have $|Z'|\leq |S(Z')|$. 
Assume now that $z\in Z'$. Let $v\in Z\backslash Z'$. As in the previous paragraph, there is a path from $z$ to $v$ that alternates between edges in $B$ and edges not in $B$. 
Let $w$ be the first vertex from $X$ in the path that is not in $Z'$, and let $w'$ be its match in $B$.
Since $w'$ can be reached directly from the vertex preceding it on the path, which belongs to $Z'$, we have $w'\in S(Z')$. 
This means that $S(Z')$ contains all vertices that are matched to $Z'\backslash\{z\}$ in $B$, along with $w'$. Hence $|S(Z')|\geq (|Z'|-1)+1 = |Z'|$.
\end{proof}

With the subroutine to compute a minimal Hall violator efficiently, we are now ready to present our main algorithm.
Recall that an envy-free assignment does not always exist; our algorithm decides whether such an assignment exists and also computes one in the case that it does.

\begin{algorithm}
\caption{Algorithm for Computing an Envy-Free Assignment}\label{alg:matching-basic}
\begin{algorithmic}[1]
\Procedure{EnvyFreeAssignment$(N, M, \text{rankings})$}{}
\State $M' \leftarrow M$
\While{$|N|\leq |M'|$}
\State Construct a bipartite graph $G=(N,M',E)$ where there is an edge from an agent to a house if and only if the house is among the most preferred houses in $M'$ for the agent. 
\If{there exists an $N$-saturating matching}
\State \Return the corresponding assignment
\Else
\State Find a minimal Hall violator $Z\subseteq N$.
\State Remove all houses adjacent to $Z$ in $G$ from $M'$.
\EndIf
\EndWhile
\State \Return null
\EndProcedure
\end{algorithmic}
\end{algorithm}

\begin{theorem}
\label{thm:envy-free-decision}
Algorithm~1 is a polynomial-time algorithm that decides whether an envy-free assignment exists and, if so, computes one such assignment.
\end{theorem}

\begin{proof}
Finding a minimal Hall violator can be done in polynomial time using Lemma~\ref{lem:hall-violator}, so each iteration of the while loop can be implemented efficiently. 
Since every iteration either returns an envy-free assignment or reduces the size of $M'$ by at least $1$, Algorithm~1 runs in polynomial time.

If the algorithm returns an assignment, every agent receives one of their most preferred houses among the assigned houses, so the assignment is envy-free. 
We will show that when the algorithm removes houses from $M'$, these houses cannot be part of any envy-free assignment. 
This will imply that if the algorithm returns null, there is indeed no envy-free assignment.

We proceed by induction on the number of rounds. Consider an arbitrary iteration of the while loop in which at least one house is removed. 
By the induction hypothesis, all houses removed in previous iterations cannot be part of an envy-free assignment. 
Let $Z$ be the minimal Hall violator that the algorithm selects in the current iteration. Assume for contradiction that a subset of houses $\emptyset\neq Y'\subseteq S(Z)$ is part of an envy-free assignment. 
Let $X'$ be the set of agents in $Z$ who only have edges to houses in $S(Z)\backslash Y'$ in $G$. Note that since $Y\neq\emptyset$, we have $X'\neq Z$. 
If $X'$ is nonempty, then since $Z$ is a minimal Hall violator, $|X'|\leq |S(X')|\leq |S(Z)\backslash Y'|$. If $X'$ is empty, $|X'|\leq|S(Z)\backslash Y'|$ holds trivially.
Since $|Z|>|S(Z)|$, it follows that $|Z\backslash X'|>|Y'|$.

By definition of $X'$, every agent in $Z\backslash X'$ has at least one most preferred house in $Y'$; since the houses in $S(Z)\backslash Y'$ are unassigned, such an agent must be assigned to a house in $Y'$. 
However, there are fewer houses in $Y'$ than agents in $Z\backslash X'$, a contradiction.
\end{proof}

We illustrate how Algorithm~1 works with two examples:
\begin{itemize}
\item Assume that there are four houses and three agents such that the agents have rankings $1\succ 4\sim 3\succ 2$, $1\succ 4\succ 2\succ 3$, and $2\succ 1\succ 3\sim 4$ over the houses.
In the first iteration, there is an edge from agent~1 to house~1, from agent~2 to house~1, and from agent 3 to house~2.
There is no $N$-saturating matching, and agents 1 and 2 form a minimal Hall violator, so house 1 is removed.
In the second iteration, there is an edge from agent~1 to houses 3 and 4, from agent~2 to house~4, and from agent~3 to house~2.
There is an $N$-saturating matching, namely the matching that assigns agent~1 to house~3, agent~2 to house~4, and agent~3 to house~2, so this assignment is returned.

\item Assume that there are four houses and three agents such that the agents have rankings $1\succ 4\succ 3\succ 2$, $1\succ 4\succ 2\succ 3$, and $2\succ 1\succ 3\sim 4$ over the houses.
In the first iteration, there is an edge from agent~1 to house~1, from agent~2 to house~1, and from agent 3 to house~2.
There is no $N$-saturating matching, and agents 1 and 2 form a minimal Hall violator, so house 1 is removed.
In the second iteration, there is an edge from agent~1 to house~4, from agent~2 to house~4, and from agent~3 to house~2.
Again, there is no $N$-saturating matching, and agents 1 and 2 form a minimal Hall violator, so house 4 is removed.
The number of agents now exceeds the number of remaining houses, so the algorithm terminates without an envy-free assignment.
\end{itemize}

Note that an assignment returned by Algorithm~1 is Pareto optimal among all envy-free assignments. 
Indeed, every agent receives one of their most preferred houses in the current iteration of the while loop, and all houses removed in previous iterations cannot be used in any envy-free assignment. 
However, envy-freeness and Pareto optimality are incompatible in general. 
To see this, consider an example with three houses and two agents such that the agents have rankings $1\succ 2\succ 3$ and $1\succ 3\succ 2$ over the houses. 
The unique envy-free assignment is to assign house~2 to agent~1 and house~3 to agent~2. On the other hand, assigning house~1 to agent~1 instead yields a Pareto improvement.

We remark that even though our Algorithm~1 may appear similar to Algorithm~2 in the paper by \cite{Segalhalevi19} at first glance, there are two crucial differences. First, since the agents have ordinal preferences over the houses in our case, our algorithm needs to redefine the bipartite graph in each iteration in order to represent these preferences; in contrast, in the setting considered by Segal-Halevi, the agents only have binary valuations and hence the graph remains unchanged throughout the execution of his algorithm. Second, and more importantly, computing a \emph{minimal} Hall violator, instead of an arbitrary one as in Segal-Halevi's case, is necessary for identifying an envy-free assignment that allocates a house to \emph{every} agent as required in our setting. To see this, consider again the first example execution of our algorithm above. In the first iteration, besides the set consisting of agents 1 and 2, the set consisting of all three agents is also a Hall violator. Removing the preferred houses 1 and 2 of the agents in this non-minimal Hall violator cannot yield an envy-free assignment, since the number of remaining houses would then be smaller than the number of agents. On the other hand, as we have already seen, the example does admit an envy-free assignment.

Next, we consider a random preference model. We assume that the agents have strict preferences over the houses, and the preference of each agent is chosen uniformly at random among all strict rankings over the houses, independently of other agents. This is equivalent to assuming that agents have cardinal utilities over the houses drawn independently from an arbitrary non-atomic distribution.\footnote{A distribution is said to be \emph{non-atomic} if it does not put positive probability on any single point.} Under this model, it is not hard to see that the probability that an envy-free assignment exists is low in the case $m=n$; indeed, an envy-free assignment exists in this case only if all agents have distinct favorite houses, a highly unlikely event. However, we show that as soon as the number of houses exceeds the number of agents by a logarithmic factor, an envy-free allocation is likely to exist.

\begin{theorem}
Let $c$ be a constant strictly greater than the base of the natural logarithm $e$.
Suppose that the agents' preferences are drawn randomly as described above, and that $m\geq cn\log n$. Then the probability that an envy-free assignment exists converges to $1$ as $n\rightarrow\infty$.
\end{theorem}

\begin{proof}
Assume without loss of generality that each agent has a cardinal utility for each house, and this utility is drawn uniformly at random from the interval $[0,1]$, independently of other pairs of agents and houses. For each house, if some agent values it at least $1-1/n$ while the remaining agents value it at most $1-1/n$, we assign it to the former agent provided that the agent has not received a house. If all agents receive a house, the resulting assignment is envy-free since all agents value their own house at least $1-1/n$ and other assigned houses at most $1-1/n$. Hence it remains to show that the probability that all agents receive a house converges to $1$.

Let $d\in(e,c)$, and fix an agent. The probability that a particular house is assigned to the agent is $\frac{1}{n}\cdot\left(1-\frac{1}{n}\right)^{n-1}$. Since $\lim_{n\rightarrow\infty}\left(1-\frac{1}{n}\right)^{n-1}=\frac{1}{e}$, we have $\frac{1}{n}\cdot\left(1-\frac{1}{n}\right)^{n-1}\geq\frac{1}{dn}$ for large enough $n$. Hence the probability that the agent does not receive a house is at most 
\begin{align*}
\left(1-\frac{1}{dn}\right)^m 
\leq \left(1-\frac{1}{dn}\right)^{cn\log n} 
\leq e^{-\frac{cn\log n}{dn}} 
= n^{-\frac{c}{d}},
\end{align*}
where the second inequality follows from $1+x\leq e^x$, which holds for every real number $x$. By union bound, the probability that some agent does not receive a house is at most $n\cdot n^{-\frac{c}{d}} = n^{1-\frac{c}{d}}$, which approaches $0$ for large $n$, completing the proof.
\end{proof}

\subsection*{Acknowledgments}

This work has been supported by the European Research Council (ERC) under grant number 639945
(ACCORD).

\bibliographystyle{named}
\bibliography{house_allocation}

\begin{thebibliography}{}

\bibitem[\protect\citeauthoryear{Abdulkadiroglu and
  S\"{o}nmez}{2003}]{AbdulkadirogluSo98}
Atila Abdulkadiroglu and Tayfun S\"{o}nmez.
\newblock Random serial dictatorship and the core from random endowments in
  house allocation problems.
\newblock {\em Econometrica}, 66(3):689--701, 2003.

\bibitem[\protect\citeauthoryear{Abraham \bgroup \em et al.\egroup
  }{2004}]{AbrahamCeMa04}
David~J. Abraham, Katar{\'{\i}}na Cechl{\'{a}}rov{\'{a}}, David Manlove, and
  Kurt Mehlhorn.
\newblock Pareto optimality in house allocation problems.
\newblock In {\em Proceedings of the 15th International Symposium on Algorithms
  and Computations}, pages 3--15, 2004.

\bibitem[\protect\citeauthoryear{Amanatidis \bgroup \em et al.\egroup
  }{2017}]{AmanatidisMaNi17}
Georgios Amanatidis, Evangelos Markakis, Afshin Nikzad, and Amin Saberi.
\newblock Approximation algorithms for computing maximin share allocations.
\newblock {\em ACM Transactions on Algorithms}, 13(4):52, 2017.

\bibitem[\protect\citeauthoryear{Beynier \bgroup \em et al.\egroup
  }{2018}]{BeynierChGo18}
Aur{\'{e}}lie Beynier, Yann Chevaleyre, Laurent Gourv{\`{e}}s, Julien Lesca,
  Nicolas Maudet, and Ana{\"{e}}lle Wilczynski.
\newblock Local envy-freeness in house allocation problems.
\newblock In {\em Proceedings of the 17th International Conference on
  Autonomous Agents and MultiAgent Systems}, pages 292--300, 2018.

\bibitem[\protect\citeauthoryear{Foley}{1967}]{Foley67}
Duncan~K. Foley.
\newblock Resource allocation and the public sector.
\newblock {\em Yale Economics Essays}, 7(1):45--98, 1967.

\bibitem[\protect\citeauthoryear{Hylland and Zeckhauser}{1979}]{HyllandZe79}
Aanund Hylland and Richard Zeckhauser.
\newblock The efficient allocation of individuals to positions.
\newblock {\em Journal of Political Economy}, 87(2):293--314, 1979.

\bibitem[\protect\citeauthoryear{Manlove}{2013}]{Manlove13}
David Manlove.
\newblock {\em Algorithmics of Matching Under Preferences}.
\newblock World Scientific, 2013.

\bibitem[\protect\citeauthoryear{Segal-Halevi}{2019}]{Segalhalevi19}
Erel Segal-Halevi.
\newblock Bipartite envy-free matching.
\newblock {\em CoRR}, abs/1901.09527, 2019.

\bibitem[\protect\citeauthoryear{Varian}{1974}]{Varian74}
Hal~R. Varian.
\newblock Equity, envy, and efficiency.
\newblock {\em Journal of Economic Theory}, 9:63--91, 1974.

\bibitem[\protect\citeauthoryear{Zhou}{1990}]{Zhou90}
Lin Zhou.
\newblock On a conjecture by {G}ale about one-sided matching problems.
\newblock {\em Journal of Economic Theory}, 52(1):123--135, 1990.

\end{thebibliography}


\end{document}